\documentclass[letter]{article}[11pt]

\makeatletter
\usepackage{hyperref}
\makeatother
\usepackage{amssymb, amsmath, setspace, epsfig,graphicx,psfrag,color,xspace,graphics}

\definecolor{Red}{rgb}{1,0,0}
\definecolor{Blue}{rgb}{0,0,1}
\definecolor{Olive}{rgb}{0.41,0.55,0.13}
\definecolor{Green}{rgb}{0,1,0}
\definecolor{MGreen}{rgb}{0,0.8,0}
\definecolor{DGreen}{rgb}{0,0.55,0}
\definecolor{Yellow}{rgb}{1,1,0}
\definecolor{Cyan}{rgb}{0,1,1}
\definecolor{Magenta}{rgb}{1,0,1}
\definecolor{Orange}{rgb}{1,.5,0}
\definecolor{Violet}{rgb}{.5,0,.5}
\definecolor{Purple}{rgb}{.75,0,.25}
\definecolor{Brown}{rgb}{.75,.5,.25}
\definecolor{Grey}{rgb}{.5,.5,.5}
\definecolor{Pink}{rgb}{1,0,1}
\definecolor{DBrown}{rgb}{.5,.34,.16}
\definecolor{Black}{rgb}{0,0,0}

\newcommand{\address}[1]{
\par {\raggedright #1
\vspace{1.4em}
\noindent\par}
}

\newtheorem{lemma}{Lemma}
\newtheorem{theorem}{Theorem}

\newenvironment{proof}[1][Proof]{\begin{trivlist}
\item[\hskip \labelsep {\bfseries #1}]}{\end{trivlist}}
\newenvironment{definition}[1][Definition]{\begin{trivlist}
\item[\hskip \labelsep {\bfseries #1}]}{\end{trivlist}}

\newcommand{\qed}{\nobreak \ifvmode \relax \else
      \ifdim\lastskip<1.5em \hskip-\lastskip
      \hskip1.5em plus0em minus0.5em \fi \nobreak
      \vrule height0.75em width0.5em depth0.25em\fi}

\newcommand{\enp} {\hfill \rule{2.2mm}{2.6mm}}

\newcommand{\mc}{\mathcal}

\newcommand{\St}{\mathrm{St}}
\newcommand{\nl}{\mathrm{null}}

\title{A rigorous analysis of the cavity equations for the minimum spanning tree}
\date{}

\begin{document}

\author{}

\author{Mohsen Bayati$^1$ \and Alfredo Braunstein$^2$ \and Riccardo Zecchina$^2$}

\maketitle

\begin{abstract}
We analyze a new general representation for the Minimum Weight Steiner
Tree (MST) problem which translates the topological connectivity
constraint into a set of local conditions which can be analyzed by the
so called cavity equations techniques.  For the limit case of the
Spanning tree we prove that the fixed point of the algorithm arising
from the cavity equations leads to the global optimum.
\end{abstract}

\address{$^1$ Microsoft Research, One Microsoft Way, 98052 Redmond, WA\\
$^2$ Politecnico di Torino, Corso Duca degli Abruzzi 24, 10129 Torino, Italy}

\maketitle

\section{Introduction}

Given a graph with positive weights on the edges, the MST problem
consists in finding a tree of minimum weight that contains a given set
of ``terminal'' vertices.  Such construction may require the inclusion
of some non\-terminal nodes which are called Steiner nodes.
Beside its practical importance in many fields, MST is a basic
optimization problem over networks which lies at the root of computer
science, being both NP-complete \cite{Karp1972} and difficult to
approximate \cite{robins2000ist}. In statistical physics the Steiner
tree problem has similarities with basic models such as polymers and
self avoiding walks with a non-trivial interplay between local an
global constraints, e.g. energy minimization versus global
connectivity.
In recent years many algorithmic results have appeared showing the
efficacy of the cavity approach for optimization and inference
problems defined over both sparse and dense random networks of
constraints
\cite{MPZ,Braunstein2006b,Frey2007,Braunstein2003,Braunstein2005,di2004wdl}.
These performances are understood in terms of factorization properties
of the Gibbs measure over ground states, which can be also seen as the
onset of correlation decay along the iterations of the cavity
equations \cite{KMRSZ2007}.
Here we make a step further in this direction by presenting evidence
for the exactness of the cavity approach for problems having an
additional rigid global constraint which couples all variables.
We show that the cavity approach can be used to derive a new algorithm \cite{BBBCRZ-PRL08}
for MST which has exact fixed points in the limit case of the
Spanning Tree.  More specifically, we show how the analysis of the
computational tree which characterizes the evolution of the so called
cavity marginals can be used to prove optimality.

\section{Definitions and Problem Statement}\label{sec:Def-Prob-Stat}

Consider an undirected simple graph $G=(V,E)$, with vertices $V =
\{1,\ldots,n\}$, and edges $E$. Let each edge $\{i,j\}$ have weight
$w_{ij}\in \mathbb{R}$. Denote {the} set of neighbors of each vertex
$i$ in $G$ by $N(i)$. Let $U$ be a subset of vertices called
\emph{terminals}.  A connected subgraph $T$ of $G$ is called
\emph{Steiner tree} if it has no cycle and contains all vertices of
$U$.  For the special case of $U=V$, the tree $T$ is called a
\emph{spanning tree}.  The set of all Steiner trees of the graph $G$
with terminals $U$ is denoted by $\St(G,U)$.

The weight of the Steiner tree $T$, denoted by $W_T$, is defined by
$W_{T}=\sum_{ij}w_{ij}1_{\{i,j\}\in T}$. The minimum weight Steiner
tree (MST), $T^*(U)$, is defined by $T^*(U)=\textrm{argmin}_{T\in
  \St(G,U)}\ W_{T}$, and for spanning trees (when $U=V$) we drop the
reference to $U$ and denote it by $T^*$.  The goal of this paper is to
present a belief propagation (BP) based algorithm for finding $T^*(U)$
and analyze it.  \emph{Throughout the paper, we will assume that
  $T^*(U)$ is unique}. { If the optimum, $T^*(U)$, is not unique then the degeneracy can be lifted by a
  small random perturbation of the weights which does not change the optimum tree.}

\section{Algorithm and Main Result}\label{sec:alg-main-res}

In this section we explain the BP algorithm for finding the minimum
weight Steiner tree.  Let us quickly explain the model. This is done
in more details in \cite{BBBCRZ-PRL08}.

\subsection{The pointer-depth model}

We model the Steiner tree problem as a rooted tree (such a
construction is often associated with the term
\emph{arborescence}). Name the vertex $1\in V$ the \emph{root}.  Then
each node $i$ is endowed with a pair of variables
$\left(p_{i},d_{i}\right)$, a pointer $p_{i}$ to some other node in
the neighborhood $N(i)$ of $i$ and a depth $d_{i}\in\left\{
1,\dots,d_{\max}\right\} $ defined as the distance from the root. Terminal
nodes (vertices in $U$) must point to some other node in the final
tree and hence $p_{i}\in N(i)$. The root node conventionally points to
itself . Non-root nodes either point to some other node in $N(i)$ if
they are part of the tree (\emph{Steiner} and \emph{terminal} nodes)
or just do not point to any node if they are not part of the tree
(allowed only for non-terminals), a fact that we represent by allowing
a ``$\nl$'' state for the pointer $p_{i}$. i.e. $p_i\in
N(i)\cup\{\nl\}$. The depth of the root is set to zero, $d_{1}=0$
while for the other nodes in the tree the depths measure the distance
from the root along the unique simple path from the node to the root.

In order to impose the global connectivity constraint for the tree we
need to impose the condition that if $p_{i}=j$ then $d_{j}=d_{i}-1$.
This condition forbids cycles and guarantees that the pointers
describe a tree. In building the BP equations, we need to introduce
the characteristic functions $f_{ij}=f_{ij}(p_i,d_i,p_j,d_j)$ which
impose such constraints over configurations of the decision variables
$\left(p_{i},d_{i}\right)$. For any edge $\left(i,j\right)$ we have
the indicator function $f_{ij}=g_{ij}g_{ji}$ where
$g_{jk}(p_j,d_j,p_k,d_k)=\left(1-\delta_{p_{k},j}\left(1-\delta_{d_{j},d_{k}-1}\right)\right)\left(1-\delta_{p_{k},j}\delta_{p_{j},\emptyset}\right)$.
Therefore any set of the decision variables
$\left\{p_{i},d_{i}\right\}_i$ that satisfies the condition
$\prod_{(i,j)\in E}f_{ij}(p_i,d_i,p_j,d_j)=1$ corresponds to a Steiner
tree in $St(G,U)$.

\subsection{BP Equations and the Algorithm}
Let us define $w_{i~\nl}=\infty$ for any { $i\notin U$}. Then the max-sum BP equations will be the followings:
\begin{align}
\psi_{j\to i}\left(d_{j},p_{j}\right)= & -w_{jp_{j}}+\sum_{k\in
  j\setminus i}\phi_{k\to
  j}\left(d_{j},p_{j}\right)\label{eq:psi}\\ \phi_{k\to
  j}\left(d_{j},p_{j}\right)= &
\max_{d_{k},p_{k}:f_{jk}\left(d_{k},p_{k},d_{j},p_{j}\right)\neq0}\psi_{k\to
  j}\left(d_{k},p_{k}\right)\label{eq:phi2}
\end{align}
On a tree $\psi_{j\to i}(d_{j},p_{j})$ can be interpreted as the
minimum cost change of removing a vertex $j$ with forced configuration
$d_{j},p_{j}$ from the subgraph with link $\left(i,j\right)$ already
removed.

On a fixed point, one computes \emph{marginals} $\psi_j$:
\begin{equation}
\psi_j\left(d_j,p_j\right)=-w_{jp_j}+\sum_{k\in j}\phi_{k\to
  j}(d_j,p_j)\label{eq:marg}\\
\end{equation}
and the BP guess of the optimum tree is given by $\arg\max \psi_j$.

{ For efficient implementation of the equations \eqref{eq:psi}-\eqref{eq:phi2} we}
introduce the variables $A_{k\to j}^{d}\equiv\max_{p_{k}\neq
  j,\nl}\psi_{k\to j}\left(d,p_{k}\right)$, $B_{k\to
  j}^{d}\equiv\psi_{k\to j}\left(d,\nl\right)$, $C_{k\to
  j}^{d}\equiv\psi_{k\to j}\left(d,j\right)$, $D_{k\to
  j}\equiv\max_{d}\max\left\{ A_{k\to j}^{d},B_{k\to j}^{d}\right\} $
and $E_{k\to j}^{d}\equiv\max\left\{ C_{k\to j}^{d+1},D_{k\to
  j}\right\} $.  This is enough to compute $\phi_{k\to
  j}\left(d_{j},p_{j}\right)=A_{k\to j}^{d_{j}-1},D_{k\to j},E_{k\to
  j}^{d_{j}}$ for $p_{j}=k$, $p_{j}=\nl$ and $p_{j}\neq k,\nl$
respectively.  Eqs. \ref{eq:psi}-\ref{eq:phi2} can then be solved by
repeated iteration of the following set of equations: \begin{eqnarray}
  A_{j\to i}^{d}\left(t+1\right) & = & \sum_{k\in j\setminus i}E_{k\to
    j}^{d}\left(t\right) + \max_{k\in i\setminus j}\left\{ A_{k\to
    j}^{d-1}\left(t\right)-E_{k\to
    j}^{d}\left(t\right)-w_{jk}\right\} \label{eq:A} \\ B_{j\to
    i}\left(t+1\right) & = & -w_{j\nl}+\sum_{k\in j\setminus i}D_{k\to
    j}\left(t\right)\label{eq:B}\\ C_{j\to i}^{d}\left(t+1\right) & =
  & -w_{ji}+\sum_{k\in j\setminus i}E_{k\to
    j}^{d}\left(t\right)\label{eq:C}\\ D_{j\to i}\left(t\right) & = &
  \max\left(\max_{d}A_{j\to i}^{d}\left(t\right),B_{j\to
    i}\left(t\right)\right)\label{eq:D}\\ E_{j\to i}^{d}\left(t\right)
  & = & \max\left(C_{j\to i}^{d+1}\left(t\right),D_{j\to
    i}\left(t\right)\right)\label{eq:E}\end{eqnarray}

Messages are initialized arbitrarily (e.g. all set to 0 at time $t=0$). Equations \ref{eq:A}-\ref{eq:E} are iterated for $t=0,1,\dots$ until $M(t)$ converges. At each iteration $t$ the estimated MST is computed as  $T(t)=\cup_{j=2}^n\left\{(j,p_j(t))\right\}$ where we define
  $p_j(t)=\arg\max_{p_j}\{\max_{d_j} \psi_j\left(t,d_j,p_j\right)\}$
  and $\psi_j\left(t,d_j,p_j\right)= \sum_{k\in j\setminus p_j}E_{k\to
    j}^{d}\left(t\right) + A_{k\to j}^{d-1}-w_{jk}$. Note that before
  convergence, $T(t)$ is not necessarily a tree.

{ One can also look at an equivalent formulation of the problem that can be constructed by
introducing a link representation of the pointer variables (introduce link variables $x_{ij}=0,\pm1$, $0$ if $i$ does not point
$j$, $1$ if $i$ points $j$ and $-1$ is $j$ points $i$).  This is a natural representation for more general versions of the Steiner tree problem but in this paper we use the pointer-depth model.}

\subsection{Main result for spanning trees}

Although iterations of equations \ref{eq:A}-\ref{eq:E} provides a distributed algorithm for
solving Steiner trees, our analysis is currently for the case of
spanning trees. Therefore throughout the rest of the paper we will
only focus on the case of $U=V$. First let us define a notion of
convergence for the algorithm.
{
\begin{definition}
Given a set of initial conditions $\left\{A_{i\to j}(0), B_{i\to
j}(0), C_{i\to j}(0), D_{i\to j}(0), E_{i\to j}(0)\right\}_{i\to j}$,
we say that the BP algorithm \emph{converges} to
$\left\{(p_i,d_i)\right\}_i$, if the decision variables converge to
$\left\{(p_i,d_i)\right\}_i$ (i.e. there exist an integer $N>0$, such
that for all $t>N$ and all $i:$ $p_i(t)=p_i, d_i(t)=d_i$).
\end{definition}
}

\begin{theorem}\label{thm:main}
If the BP algorithm converges to $\left\{(p_i,d_i)\right\}_i$, then
the set of the edges $\{(i,p_i)\}_i$ is the minimum spanning tree
$T^*$.
\end{theorem}
\paragraph{Note 1.}
For Theorem \ref{thm:main} to hold we only need the equalities
$p_i(t)=p_i, d_i(t)=d_i$ to hold for $N<t\leq N+2d_{\max}+1$.

{
\paragraph{Note 2.}
There are examples for which this BP algorithm does not converge and one needs to use some heuristics to make it converge \cite{Braunstein2006b}. To the best of our knowledge there is no rigorous analysis of these heuristics in the literature.
}

\section{Analysis}
Before proving the Theorem \ref{thm:main} we quickly review the notion
of \emph{computation tree}. Computation trees have been used in most
of the previous analysis of the BP algorithms; see \cite{Bayati2007,BBCZ-JSTAT08}
for a list those works.

\subsection{Computation Tree.}
For any $i\in V$, let $T_{i}^t$ be the $t$-level computation tree
corresponding to $i$, defined as follows: $T_{i}^{t}$ is a weighted
tree of height $t+1$, {rooted at $i$}. All tree-nodes have labels from
the set $\{1,\ldots,n\}$ according to the following recursive rules:

(a) {The root} has label $i$.

(b) The {set of labels of the} $deg_G(i) $ children of the root is
equal to $N(i)$.

(c) {If $s$ is a non-leaf node whose parent has label $r$, then the
  set of labels of its children is $N(s)\backslash\{r\}$.}

\textbf{Notation.} We denote a vertex $u$ of the computation tree by
$[u,i]$ if it has label $i$. We also denote root of the computation
tree $T_i^t$ by $[root,i]$.

Similar to the pointer-depth model in graph $G$, we assign to each
non-leaf vertex $[v,j]$ of $T_{i}^{t}$ two decision variables
$(p_{v},d_{v})$ with $p_{v}\in N([v,j])$, and
$d_{v}\in\{1,\ldots,d_{\max}\}$.  We call such an assignment \emph{valid} if
the following constraints are satisfied:

(a) If for two neighbors $[u,j],[v,k]$ in $T_i^t$, $p_v=[u,j]$ then $d_u=d_v+1$.

(b) For any vertex $[u,j]$ in $T_i^t$ whose label is the same as the root in $G$ (i.e. $j=1$), then $d_u=0$.

Now for any valid assignment, the subtree $\mc{T}=\{([v,j],p_{v})\}$
is called an \emph{oriented spanning tree} of the computation
tree. Figure \ref{fig:one} shows a graph with one of its computation
trees, and an oriented spanning tree on it.  Denote the minimum weight
oriented spanning tree (MWOST) of the computation tree $T_{i}^{t}$ by
$\mc{T}^*(T_{i}^{t})$. Similar argument as in \cite{Bayati2007} shows
that iterations of Eqs.\ref{eq:A}-\ref{eq:E} can be seen as a dynamic programming procedure that finds the MWOST over the computation tree. And Lemma
\ref{cor:bp-solves-tree} that comes next without proof is analogues to
the Corollary 1 { from} \cite{Bayati2007}.
\begin{figure}
\centering
      \includegraphics[width=0.85\linewidth,angle=0]{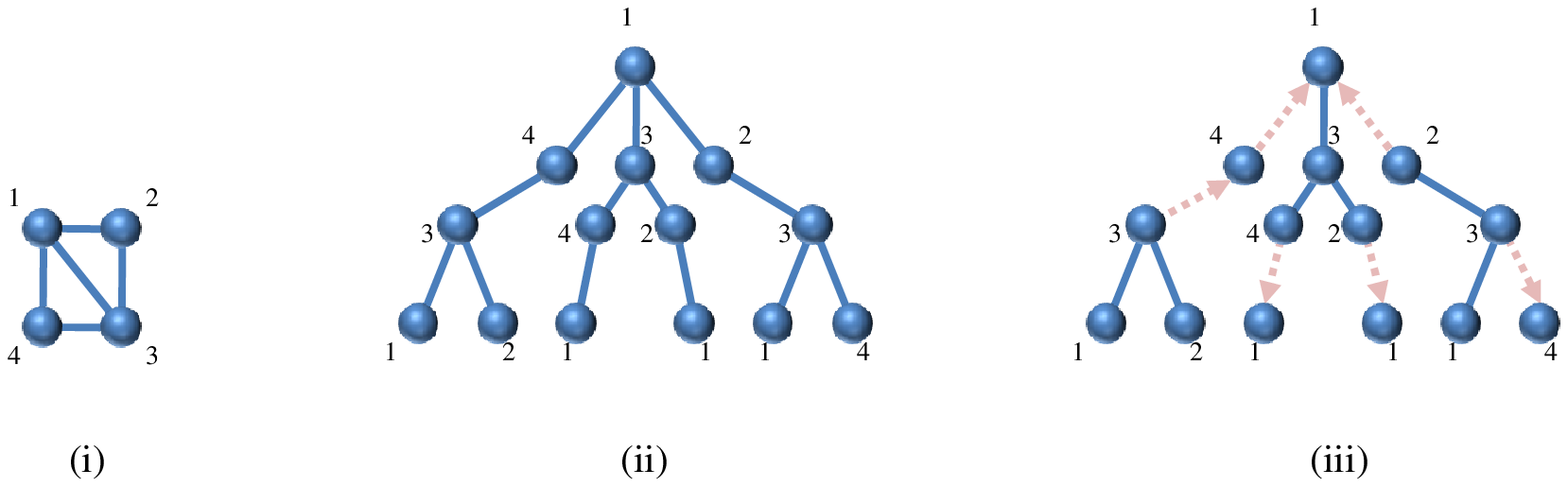}
      \caption{(i) shows a graph with 4 vertices, (ii) represents the
        computation tree $T_1^2$ for it, and (iii) shows an oriented
        spanning tree on the computation tree.}
      \label{fig:one}
\end{figure}

\begin{lemma}\label{cor:bp-solves-tree}
The BP algorithm that is initialized with zero messages, solves the
MWOST problem on the computation tree. In particular, for each vertex
$i$ of $G$ the decision variables $(p_i(t),d_i(t))$ are exactly equal
to the decision variables $(p_i,d_i)$ corresponding to the vertex
$[root, i]$ in $\mc{T}^*(T_{ i}^{t})$.
\end{lemma}

\paragraph{Note 3.}
Lemma \ref{cor:bp-solves-tree} can be generalized to any
\emph{unbalanced computation tree} (a tree that is obtained from
$T_i^t$ by removing a subset of vertices and all of their descendants)
as well.  For an unbalanced tree, there is a unique set of BP initial
conditions that should be used instead of the zero messages. { Lemma \ref{cor:bp-solves-tree} holds for
any model where BP is used and does not depend to the problems studied in this paper (See
\cite{Wei00}, \cite{WeF01}, and \cite{WeF01b} for more details).}

{
\paragraph{Note 4.}
We would like to point out that the main result holds for the BP algorithm with any initial condition and
we assume zero initial condition just to simplify the calculations.  For arbitrary initial condition,
the BP algorithm runs over a slightly modified computation tree. {The new computation
tree is almost} the same computation tree as $T_{ i}^t$, except {that} the leaf edges of
the tree have arbitrary weights and not $w_{ij}$'s from $G$.
}

\subsection{Proof of the main result}

Proof consists of two parts. First we will show { that in} case of
convergence, the estimated MST $T(t)$ is a spanning
tree. Next we will prove that this limit is in fact the minimum
spanning tree.

\subsubsection{Limit is a spanning tree}
First we will show that the limit of the BP algorithm is a spanning
tree.
\begin{lemma}\label{lem:unbalanced-tree}
If the BP algorithm converges to $\left\{(p_i,d_i)\right\}_i$, then
the set of edges $\{(i,p_i)\}_i$ is a spanning tree of $G$.
\end{lemma}
\begin{proof}
Let us denote the set of edges $\{(i,p_i)\}_i$ by $T$. Note that from
Lemma \ref{cor:bp-solves-tree} (and the note after that) we obtain the
following.  Since BP algorithm converges to
$\left\{(p_i,d_i)\right\}_i$, then for any vertex $i$ and any radius
$r_i$ one can find a large enough computation tree with root
$[root,i]$ such that in the MWOST of that computation tree, all of the
vertices within distance $r_i$ of the root have decision variables
that are dictated by $\{(i,p_i)\}_i$.  In other words there exist a
number $N_i$ such that in the MWOST of the computation tree
$T_i^{N_i}$, any vertex $[u,j]$ with distance less than $r_i$ from
$[root,i]$ has $d_u=d_j$ and $p_u=[*,p_j]\in N([u,j])$.

Now consider the MWOST $\mc{T}^*(T_{ i}^{N_i})$.  It consists of many
connected pieces. Let $A_i$ be the connected component of
$\mc{T}^*(T_{ i}^{N_i})$ that contains the $[root,i]$.  Note that each
edge of $A_i$ corresponds to some $([u,j],p_u)$ by definition. We list
and prove a few properties about the subtree $A_i$:

\begin{itemize}
\item[(i)] \textbf{$A_i$ has bounded radius.} All vertices of $A_i$
  are within distance at most $2d_{\max}$ from $[root,i]$.\\ \emph{Proof.}
  Consider the unique path $P([u,j],[root,i])$ in $A_i$ that connect a
  vertex $[u,j]$ to the $[root,i]$. The depth variable along the path
  $P([u,j],[root,i])$ either always increases by one (thus
  $|P([u,j],[root,i])|\leq d_{\max}$) or it always decreases by 1 till it
  reaches zero and then increases by 1 up to $d_{\max}$ (or
  $|P([u,j],[root,i])|\leq 2d_{\max}$).

\item[(ii)] \textbf{$A_i$ has no duplicate vertex.} No two vertices of
  $A_i$ have the same labels from the set $V$.  That means no two
  vertices of the form $[u,j]$ and $[v,j]$ belong to
  $A_i$.\\ \emph{Proof.} Assume the contrary, then let $[u,j]$ and
  $[v,j]$ be two such vertices in $A_i$ which have the smallest depth
  variables $d_u=d_v$ (note that by property (i) both $[u,j]$ and
  $[v,j]$ are within distance $2d_{\max}$ of $[root,i]$ which shows
  $d_u=d_v=d_j$).  First assume $j\neq 1$. Consider the vertices of
  the computation tree that are pointed to by $[u,j]$ and $[v,j]$
  (i.e. $p_u=[u',p_j]$ and $p_v=[v',p_j]$).  By design both $[u',p_j]$
  and $[v',p_j]$ belong to $A_i$ since they are connected to $[u,j]$
  and $[v,j]$ respectively, and $d_{u'}=d_u-1=d_j-1$,
  $d_{v'}=d_v-1=d_j-1$. Hence we should have $[u',p_j]=[v',p_j]$ (by
  definition of $[u,j]$ and $[v,j]$ that have smallest value for
  $d_u=d_v$). But this means the vertex $[u',p_j]$ of the computation
  tree has two distinct neighbors $[u,j]$ and $[v,j]$ with the same
  label which is a contradiction because the computation tree and $G$
  have the same local structure at any non-leaf vertex.  The case
  $j=1$ is trivially impossible since the depth variable along the
  path between $[u,j]$ and $[v,j]$ should go from zero to zero.

\item[(iii)] \textbf{$A_i$ has all labels from $V$.} First note that
  $A_i$ has a vertex with label $1$. Because starting form $i$ and
  following the pointers the depth variable is decreasing and it
  becomes zero at some point. That vertex which has depth zero is in
  $A_i$ and has to have label $1$.  Now we show that for any $j\in V$
  there exist a vertex $[u,j]\in A_i$.  Consider the sequence $S=j,
  p_j, p_{p_j}, p_{p_{p_j}}, \ldots$. This sequence has to stop at $1$
  since the depth variable for elements of the sequence is strictly decreasing. So it eventually intersects
  labels that appear in $A_i$. Consider the first time that the
  intersection happens (for an element $[u,k]$ of $A_i$ we have $k\in
  S$). If $\ell$ is the element before $k$ in $S$
  (i.e. $p_\ell=k$). We prove that $\ell$ is also a label in
  $A_i$. This is because $[u,k]$ has the same local structure in the
  computation tree as $k$ in $G$ and $\ell$ is a neighbor of $k$ in
  $G$. Thus there exist a $[v,\ell]\in N([u,k])$ and the distance
  between $[v,\ell]$ and $[root,i]$ is at most $2d_{\max}+1$. So $p_v=p_\ell$
  and $d_v=d_\ell$.  This means that $[v,\ell]$ is connected to
  $[u,k]$ and hence is in $A_i$. Repeating the process, we obtain that
  $j$ is a label in $A_i$.
\end{itemize}
Properties (i)-(iii) show that under $[u,j]\to j$ the tree $A_i$ is
isomorphic to $T=\{(p_i, d_i)\}_i$ and therefore $T$ is a spanning
tree of $G$.  \enp
\end{proof}

\subsubsection{Limit is the minimum weight spanning tree}

To prove that the set $T=\{(p_i,d_i)\}_i$ is the minimum spanning tree
we assume the contrary ($T\neq T^*$). Then we will construct an
oriented spanning tree $\mc{T}(T_i^{N_i})$ that has less weight than
$\mc{T}^*(T_i^{N_i})$ which is a contradiction.

For our proof, we need to give a quick review of Prim's well-known algorithm \cite{Prim1957} for finding the minimum
spanning tree of the graph $G$. The algorithm continuously increases the size of a tree starting with a single vertex until it spans all the vertices.
It starts from an initial subtree $T_0$ of $G$ that contains a single vertex. Then for any $r=0,1,\ldots,n-2$ the following step is repeated: Find the minimum weight edge $(u,v)$ that connects $T_r$ to $G\backslash T_r$ and set $T_{r+1}=T_r\cup \{(u,v)\}$. The tree $T_{n-1}$ is the minimum spanning tree.

Assume that the Prim algorithm starts
with the vertex $1$. Let $e_1,e_2,\ldots,e_{n-1}$ be the order of the
edges that are added during the algorithm. That is
$T_r=\{e_1,e_2,\ldots,e_{r-1}\}$.  Now let $e_k$ be the first edge
that does not belong to $T$. The subgraph $T\cup\{e_k\}$ has a
cycle. Thus it has has an edge $e$ in $T$ that connects $T_{k-1}$ to
outside of $T_{k-1}$. By Prim's algorithm, $w(e)<w(e_k)$. The
inequality is strict since $T^*$ is unique.

Let $T'=(T\cup\{e_k\})\backslash\{e\}$. It is not hard to see that
$T'$ is also a spanning tree of $G$ and $w(T')<w(T)$. Consider the
pointer-depth representation for the tree $T'$ and denote the
corresponding decision variables by $\{(p_i',d_i')\}_i$. Let also
$(x,p_x')$ corresponds to the edge $e$ in this new pointer-depth
representation. Since $1\in T_k\subset T\cap T'$ then for any $i\in
T_k$ we have $(p_i,d_i)=(p_i',d_i')$.

Now we consider the oriented spanning tree
$\mc{T}(T_i^{N_i})$. Similar to the previous section, let $A_i$ be the
connected component of $\mc{T}^*(T_{i}^{N_i})$ that contains the
$[root,i]$. Let $[u,x]\in A_i$ be the unique vertex that has label
$x$. We will change the decision variables of any vertex $[v,j]$ of
$A_i$ from $(p_v,d_j)$ to $(p_v',d_j')$ where $p_v'$ is the unique
vertex in $N([v,j])$ that has label $p_j'$. Denote the new subgraph of
the computation tree by $\mc{T}'(T_{i}^{N_i})$. Clearly $w(\mc{T}'(T_{i}^{N_i}))<w(\mc{T}(T_{i}^{N_i}))$.  Now we only need to show
that $\mc{T}'(T_{i}^{N_i})$ is an oriented spanning tree of the
computation tree to achieve a contradiction.

Since $T'\backslash T=\{e\}$, therefore we only need to check that
local constraints at edge $([u,x],p_u')$ of $\mc{T}'(T_{i}^{N_i})$
satisfy the ones of an oriented spanning tree. Note that all neighbors
of the vertex $[u,x]$ are within the distance $2d_{\max}+1$ of
$[root,i]$. Thus if, $p_u'=[v,p_x']$ then $(p_v,d_v)$ will be equal to
$([*,p_{p_x'}],d_{p_x'}')$. On the other hand $p_x'$ is a vertex in
$T_k$ and for all vertices of $T_k$ the decision variables $(p,d)$ and
$(p',d')$ are the same. Thus $(p_u',d_u')$, $(p_v',d_v')$ will satisfy
the local constraints since $(p_x',d_x')$, $(p_{p_x'},d_{p_x'})$
satisfy the same constraint in $T'$. Therefore we obtained a new
oriented spanning tree of the computation tree which has weight less
than the optimum, $\mc{T}^*(T_{ i}^{N_i})$, which is a
contradiction. So the assumption $T\neq T^*$ was incorrect.  \enp

\section{Acknowledgements}\label{sec:ack}
Mohsen Bayati acknowledges the support of the Theory Group at Microsoft Research and Microsoft Technical Computing Initiative.

\bibliographystyle{unsrt} \bibliography{all}

\begin{thebibliography}{10}

\bibitem{Karp1972}
R.M. Karp.
\newblock {Reducibility among combinatorial problems}.
\newblock {\em Complexity of Computer Computations} {\bf 43}, 85--103 (1972)

\bibitem{robins2000ist}
G.~Robins and A.~Zelikovsky.
\newblock {Improved Steiner tree approximation in graphs}.
\newblock {\em Proceedings of the eleventh annual ACM-SIAM symposium on
  Discrete algorithms}, 770--779 (2000)

\bibitem{MPZ}
Marc Mezard, Giorgio Parisi, and Riccardo Zecchina.
\newblock Analytic and algorithmic solution of random satisfiability problems.
\newblock {\em Science} {\bf 297}, 812 (2002)

\bibitem{Braunstein2006b}
Alfredo Braunstein and Riccardo Zecchina.
\newblock Learning by message-passing in networks of discrete synapses.
\newblock {\em Phys. Rev. Lett.} {\bf 96}, 030201 (2006)

\bibitem{Frey2007}
B.J. Frey and D.~Dueck.
\newblock {Clustering by Passing Messages Between Data Points}.
\newblock {\em Science} {\bf 315}, 972 (2007)

\bibitem{Braunstein2003}
Alfredo Braunstein, Roberto Mulet, Andrea Pagnani, Martin Weigt, and Riccardo
  Zecchina.
\newblock Polynomial iterative algorithms for coloring and analyzing random
  graphs.
\newblock {\em Phys. Rev. E} {\bf 68}, 036702 (2003)

\bibitem{Braunstein2005}
Alfredo Braunstein, Marc Mezard, and Riccardo Zecchina.
\newblock Survey propagation: an algorithm for satisfiability.
\newblock {\em Random Structures and Algorithms} {\bf  27}, 201--226 (2005)

\bibitem{di2004wdl}
C.~Di, A.~Montanari, and R.~Urbanke.
\newblock {Weight distributions of LDPC code ensembles: combinatorics meets
  statistical physics}.
\newblock {\em Proceedings. International Symposium on Information Theory. ISIT
  2004.}, 102 (2004)

\bibitem{KMRSZ2007}
F.~Krzakala, A.~Montanari, F.~Ricci-Tersenghi, G.~Semerjian, and L.~Zdeborova.
\newblock Gibbs states and the set of solutions of random constraint
  satisfaction problems.
\newblock {\em Proceedings of the National Academy of Sciences} {\bf 104}, 10318
(2007)

\bibitem{BBBCRZ-PRL08}
M.~Bayati, C.~Borgs, A.~Braunstein, J.~Chayes, A.~Ramezanpour, and R.~Zecchina.
\newblock Statistical mechanics of steiner trees.
\newblock {\em Phys. Rev. Lett.} {\bf 101}, 037208 (2008)

\bibitem{Bayati2007}
M.~Bayati, C.~Borgs, J.~Chayes, and R.~Zecchina.
\newblock Belief-propagation for weighted b-matchings on arbitrary graphs and
  its relation to linear programs with integer solutions.
\newblock {\em eprint arXiv: 0709.1190}, Sep. 2007.

\bibitem{BBCZ-JSTAT08}
M~Bayati, C.~Borgs, J.~Chayes, and R.~Zecchina.
\newblock On the exactness of the cavity method for weighted b-matchings on
  arbitrary graphs and its relation to linear programs.
\newblock {\em J. Stat. Mech.}  L06001, 2008.

\bibitem{Wei00}
Y.~Weiss.
\newblock Correctness of local probability propagation in graphical models with
  loops.
\newblock {\em Neural Comput.} {\bf 12}, 1 (2002)

\bibitem{WeF01}
Y.~Weiss and W.~Freeman.
\newblock Correctness of belief propagation in gaussian graphical models of
  arbitrary topology.
\newblock {\em Neural Comput.} {\bf 13}, 2173 (2001)

\bibitem{WeF01b}
Y.~Weiss and W.~Freeman.
\newblock On the optimality of solutions of the max--product
  belief--propagation algorithm in arbitrary graphs.
\newblock {\em IEEE Trans. Info. Theory} {\bf 47}, 736-744 (2001)

\bibitem{Prim1957}
R.~Prim.
\newblock {Shortest Connection Matrix Network and Some Generalisations}.
\newblock {Bell System Technical Journal} {\bf 36}, 1389--1401 (1957)

\end{thebibliography}

\end{document}